\documentclass[12pt]{article}

\usepackage{CJKutf8}
\usepackage{cite}
\usepackage{graphicx}
\usepackage{amsmath}
\usepackage{amsthm}

\newtheorem{lemma}{Lemma}

\usepackage{geometry}
\usepackage{framed}
\usepackage{multirow}
\usepackage{rotating}
\usepackage{amsfonts}
\usepackage{longtable}
\usepackage{array} %%

\usepackage{fancyhdr}
\begin{document}

\title{Guaranteed Fast Implementation of the Split Covariance Intersection Filter: Nested Newton Method Thanks to the Fourth-Order Convexity of $w$-Optimization}
\date{}

\author{Hao Li 
\thanks{Namely \begin{CJK}{UTF8}{gbsn}李颢\end{CJK} .}
}

\maketitle

\begin{abstract}
The split covariance intersection filter (Split CIF) is a useful tool for general data fusion and has the potential to be applied in a variety of engineering tasks. The $w$-optimization problem involved in the Split CIF concerns the performance and implementation efficiency of the Split CIF. It is known that the $w$-optimization problem enjoys the desirable property of convexity (or more clearly, the second-order convexity in this paper's context). This paper proves that the $w$-optimization problem further enjoys a more desirable property namely the fourth-order convexity, thanks to which a guaranteed fast implementation of the Split CIF can be realized. The new implementation is coined as the nested Newton method, which is also presented in this paper.
\end{abstract}

\section{Introduction}

The split covariance intersection filter (or Split CIF for short) \cite{Julier2001, Li2013a, Cros2025} is a useful tool for general data fusion
\footnote{The paper \cite{Julier2001}, which presents the Split CIF heuristically without theoretical analysis, originally coined it simply as ``split covariance intersection''. The paper \cite{Li2013a} provides a theoretical foundation for the Split CIF besides its original heuristic version. In \cite{Li2013a}, the term ``filter'' is added to form an analogy of the Split CIF to the well-known Kalman filter. Although the Split CIF is called ``filter'', it is not limited to temporal recursive estimation but can be used as a pure data fusion method besides the filtering sense, just as the Kalman filter can also be treated as a data fusion method.}.
As explained in \cite{Li2013a}, the Split CIF can reasonably handle both known independent information and unknown correlated information in source data, and has the potential to be applied in a variety of engineering tasks.
Representative application of the Split CIF can be in cooperative intelligent systems \cite{Li2013b, Wanasinghe2014, Pierre2018, ChenX2020, Li2022TITS, Li2024ITS, Li2025VTT} and in single intelligent system operation \cite{Li2013d, Allig2022, Li2022RAL, Li2023AprilTagNavigation, Li2024IV} as well.

An indispensable optimization step involved in the Split CIF, namely \textbf{$w$-optimization} (refer to Section \ref{sec:woptprob} for details of the problem statement), concerns the performance and implementation efficiency of the Split CIF. Convexity is always a desired property for optimization problems as it facilitates optimization considerably. Implementation of the Split CIF can also rely on convexity of the $w$-optimization problem which is theoretically guaranteed. The proof details for the convexity are provided in the author's works \cite{Li2022FARET_2, Li2022FARET_1}, whereas an online preprint excerpt can also be found in \cite{Li2021WOpt}.

The $w$-optimization problem of the Split CIF enjoys not only the desirable property of convexity (or more clearly, the second-order convexity in this paper's context) but also a more desirable property namely the fourth-order convexity. Thanks to the fourth-order convexity of the $w$-optimization problem, a guaranteed fast implementation of the Split CIF, which is coined as the nested Newton method, can be realized. This paper proves the fourth-order convexity of the $w$-optimization problem and introduces the nested Newton method based on the fourth-order convexity.

\section{The $w$-optimization problem} \label{sec:woptprob}

Matrices mentioned in this paper are symmetric matrices by default. Given matrices $\mathbf{P}_{1d}$, $\mathbf{P}_{1i}$, $\mathbf{P}_{2d}$, and $\mathbf{P}_{2i}$ that are positive semi-definite, i.e. $\mathbf{P}_{1d} \geq \mathbf{0}$, $\mathbf{P}_{1i} \geq \mathbf{0}$, $\mathbf{P}_{2d} \geq \mathbf{0}$, $\mathbf{P}_{2i} \geq \mathbf{0}$; denotations $\mathbf{P}_{1d}$, $\mathbf{P}_{1i}$, $\mathbf{P}_{2d}$, and $\mathbf{P}_{2i}$ are used for presentation of the Split CIF in \cite{Li2013a}. For $w \in [0,1]$, define
\begin{align} \label{eq:defineP}
\mathbf{P}_{1}(w) &= \mathbf{P}_{1d}/w + \mathbf{P}_{1i} \nonumber \\
\mathbf{P}_{2}(w) &= \mathbf{P}_{2d}/(1-w) + \mathbf{P}_{2i} \nonumber \\
\mathbf{P}(w) &= (\mathbf{P}_{1}(w)^{-1} + \mathbf{P}_{2}(w)^{-1})^{-1}
\end{align}
When $w=0$ or $w=1$, $\mathbf{P}(w)$ denotes the limit value as $w \to 0$ or $w \to 1$ respectively. For $w \in (0,1)$, we further assume that $\mathbf{P}_{1}(w)$ and $\mathbf{P}_{2}(w)$ are positive definite i.e. $\mathbf{P}_{1}(w)>0$, $\mathbf{P}_{2}(w)>0$. This fair assumption is well rooted in practical applications where $\mathbf{P}_{1}(w)$ and $\mathbf{P}_{2}(w)$ normally correspond to covariances of certain estimates and hence are always positive definite. With this assumption, we naturally have $\mathbf{P}(w)>0$.

The $w$-optimization problem involved in the Split CIF \cite{Li2013a} can be formalized as
\begin{equation}  \label{eq:wopt_org}
w = \arg \min_{w \in [0,1]} \det(\mathbf{P}(w))
\end{equation}
or equivalently
\begin{equation}  \label{eq:wopt}
w = \arg \min_{w \in [0,1]} \ln \det(\mathbf{P}(w)).
\end{equation}
It is worth noting that the $w$-optimization problem can be generally formalized as
\begin{equation}  \label{eq:wopt_general}
w = \arg \min_{w \in [0,1]} \|| \mathbf{P}(w) \||,
\end{equation}
where the abused matrix-norm-style notation $\|| \cdot \||$ is in fact viewed as a generic positive definite real-value matrix function rather than certain strict matrix norm, only if it serves as an appropriate metric that characterizes the ``size'' of the matrix $\mathbf{P}(w)$. It can indeed be a matrix norm, yet two commonly adopted choices are the determinant function $\det(\cdot)$ and the trace function $tr \{ \cdot \}$ respectively. 

According to his field experiences, the author recommends the determinant function because it saves the not-easy-to-tune step of normalization which is inevitable if the trace function is adopted. This is why the determinant function $\det(\cdot)$ is adopted by default in the $w$-optimization problem. Besides, concerning the two equivalent formalisms (\ref{eq:wopt_org}) and (\ref{eq:wopt}), the latter namely (\ref{eq:wopt}) is preferred as it facilitates analysis and is numerically more desirable. Therefore, whenever the $w$-optimization problem is mentioned throughout this paper, it refers to (\ref{eq:wopt}) by default.

The conventional-sense convexity of the $w$-optimization problem is equivalent to the following inequality
\begin{equation} \label{eq:convexC2}
\frac{d^2}{dw^2} \ln \det(\mathbf{P}(w)) \geq 0
\end{equation}
holding true for $w \in (0,1)$.

Besides, consider an even stronger (and more desirable) version of convexity of the $w$-optimization problem, which is equivalent to the following inequality
\begin{equation} \label{eq:convexC4}
\frac{d^4}{dw^4} \ln \det(\mathbf{P}(w)) \geq 0
\end{equation}
also holding true besides (\ref{eq:convexC2}). The stronger convexity (\ref{eq:convexC4}) is referred to as the \textbf{fourth-order convexity}. In other words, the fourth-order derivative of the objective function is always non-negative. For a generic objective function $f(x)$, if its $m$-th-order derivative is always non-negative, namely
\begin{equation}  \label{eq:m_order_convexity}
\frac{d^m}{dx^m} f(x) \geq 0,
\end{equation}
then it is said to have the \textbf{$m$-th-order convexity}. According to (\ref{eq:m_order_convexity}), the conventional-sense convexity is right the second-order convexity.

It is worth noting that an existing concept similar to the $m$-th-order convexity defined in (\ref{eq:m_order_convexity}) is the ``$(m-1)$-hyperconvexity'' used in \cite{Vatsala2006}. However, the author refrains from borrowing this already existing concept from \cite{Vatsala2006}, for sake of not arousing confusion, because the concept ``\{number\}-hyperconvexity'' (for example, $0$-hyperconvexity, $1$-hyperconvexity, $2$-hyperconvexity, and so on) has long since been already used to mean a totally different thing for metric geometry in mathematics literature \cite{Aronszajn1956, Rabier1985}.

\section{The fourth-order convexity of $w$-optimization} \label{sec:fourth_order_convexity}

As mentioned previously, the proof details for the second-order convexity are provided in \cite{Li2022FARET_2, Li2022FARET_1}. This section provides the proof details for the fourth-order convexity. In other words, this section proves (\ref{eq:convexC4}). How to take advantage of the four-order convexity to realize the nested Newton method is postponed to Section \ref{sec:nested_newton}.

For denotation conciseness in the following proof, we omit explicit writing of ``$(w)$'' for $w$-parametrized variables; for example, we denote above mentioned $\mathbf{P}_{1}(w)$, $\mathbf{P}_{2}(w)$, and $\mathbf{P}(w)$ simply as $\mathbf{P}_{1}$, $\mathbf{P}_{2}$, and $\mathbf{P}$. Besides, it is taken for granted that all $w$-parametrized variables in following theoretical analysis are arbitrary-order differentiable, as relevant $w$-parametrized variables in practical applications are indeed so.

\begin{lemma} \label{lm:diff1}
Given a $w$-parameterized matrix $\mathbf{M}(w)$ satisfying $\mathbf{M}(w)>0$, we have
\begin{equation*}
\frac{d}{dw} \ln \det(\mathbf{M}) = tr\{\mathbf{M}^{-1} \frac{d \mathbf{M}}{dw} \}.
\end{equation*}
\end{lemma}
\begin{proof}
According to the Jacobi's formula \cite{Horn1991}
\begin{equation*}
\frac{d}{dw} \det(\mathbf{M}) = \det(\mathbf{M}) tr\{\mathbf{M}^{-1} \frac{d \mathbf{M}}{dw} \}.
\end{equation*}
Thus we have
\begin{equation*}
\frac{d}{dw} \ln \det(\mathbf{M}) = \frac{1}{\det(\mathbf{M})} \frac{d}{dw} \det(\mathbf{M}) = tr\{\mathbf{M}^{-1} \frac{d \mathbf{M}}{dw} \}.
\end{equation*}
\end{proof}

\begin{lemma}  \label{lm:diff_inv}
Given a matrix $\mathbf{M}(w)$, the derivative of its inverse is computed as
\begin{equation*}
\frac{d \mathbf{M}^{-1}}{dw} = -\mathbf{M}^{-1} \frac{d \mathbf{M}}{dw} \mathbf{M}^{-1}.
\end{equation*}
\end{lemma}
Validity of \textbf{Lemma} \ref{lm:diff_inv} can be easily known according to the following equality
\begin{align*}
\mathbf{0} = \frac{d \mathbf{I}}{d w} = \frac{d (\mathbf{M}^{-1} \mathbf{M})}{d w} = \frac{d \mathbf{M}^{-1}}{d w} \mathbf{M} + \mathbf{M}^{-1} \frac{d \mathbf{M}}{d w}.
\end{align*}

\begin{lemma} \label{lm:lndetP}
Note (\ref{eq:defineP}) and we have
\begin{align*}
\ln \det(\mathbf{P}) = \ln \det(\mathbf{P}_{1}) + \ln \det(\mathbf{P}_{2}) - \ln \det(\mathbf{P}_{1} + \mathbf{P}_{2}).
\end{align*}
\end{lemma}
\begin{proof}
From (\ref{eq:defineP}) we can derive
\begin{align*}
\mathbf{P} = (\mathbf{P}_{1}^{-1} + \mathbf{P}_{2}^{-1})^{-1} = (\mathbf{P}_{2}^{-1} (\mathbf{P}_{2} + \mathbf{P}_{1}) \mathbf{P}_{1}^{-1})^{-1} = \mathbf{P}_{1} (\mathbf{P}_{1} + \mathbf{P}_{2})^{-1} \mathbf{P}_{2},
\end{align*}
which further leads to
\begin{align*}
\det(\mathbf{P}) &= \det(\mathbf{P}_{1} (\mathbf{P}_{1} + \mathbf{P}_{2})^{-1} \mathbf{P}_{2}) = \det(\mathbf{P}_{1}) \cdot \det(\mathbf{P}_{2}) \cdot \det((\mathbf{P}_{1} + \mathbf{P}_{2})^{-1}) \\
  &= \frac{\det(\mathbf{P}_{1}) \cdot \det(\mathbf{P}_{2})}{\det(\mathbf{P}_{1} + \mathbf{P}_{2})}.
\end{align*}
So
\begin{align*}
\ln \det(\mathbf{P}) = \ln \frac{\det(\mathbf{P}_{1}) \cdot \det(\mathbf{P}_{2})}{\det(\mathbf{P}_{1} + \mathbf{P}_{2})} = \ln \det(\mathbf{P}_{1}) + \ln \det(\mathbf{P}_{2}) - \ln \det(\mathbf{P}_{1} + \mathbf{P}_{2}).
\end{align*}
\end{proof}

\begin{lemma} \label{lm:trcyclic}
Given two matrices $\mathbf{M}_1$ and $\mathbf{M}_2$ whose dimensions are consistent with each other for multiplication $\mathbf{M}_1 \mathbf{M}_2$ and $\mathbf{M}_2 \mathbf{M}_1$, we have $tr\{\mathbf{M}_1 \mathbf{M}_2\} = tr\{\mathbf{M}_2 \mathbf{M}_1\}$.
\end{lemma}
The proof for \textbf{Lemma}.\ref{lm:trcyclic} can be found in \cite{Horn1990}. More generally, given matrices $\mathbf{M}_1$, $\mathbf{M}_2$, ..., $\mathbf{M}_k$, we have 
\begin{align*}
&tr\{\mathbf{M}_1 \mathbf{M}_2 ... \mathbf{M}_k\} = tr\{\mathbf{M}_2 \mathbf{M}_3 ... \mathbf{M}_k \mathbf{M}_1\} \\
&~~~~= ... = tr\{\mathbf{M}_k \mathbf{M}_1 ... \mathbf{M}_{k-2} \mathbf{M}_{k-1}\},
\end{align*}
which is called \textit{cyclic property} of trace operation.

\subsection{Derivatives of the objective function}

Following \textbf{Lemma} \ref{lm:diff1} to \textbf{Lemma} \ref{lm:lndetP} and using \textbf{Lemma} \ref{lm:trcyclic} (the cyclic property of trace operation) when necessary in following derivation, we can compute the first-order, second-order, third-order, and fourth-order derivatives of the objective function $\ln \det(\mathbf{P}(w))$ consecutively as
\begin{align}  \label{eq:lndetP_1stDif}
\frac{d}{d w} \ln \det \mathbf{P} = tr\{ \mathbf{P}_1^{-1} \frac{d \mathbf{P}_1}{dw} + \mathbf{P}_2^{-1} \frac{d \mathbf{P}_2}{dw} - (\mathbf{P}_1+\mathbf{P}_2)^{-1} \frac{d (\mathbf{P}_1+\mathbf{P}_2)}{dw} \},  \qquad \qquad \qquad \quad
\end{align}

\begin{align} \label{eq:lndetP_2ndDif}
& \frac{d^2}{dw^2} \ln \det \mathbf{P} = tr\{-\mathbf{P}_1^{-1} \frac{d \mathbf{P}_1}{dw} \mathbf{P}_1^{-1} \frac{d \mathbf{P}_1}{dw} + \mathbf{P}_1^{-1} \frac{d^2 \mathbf{P}_1}{dw^2} -\mathbf{P}_2^{-1} \frac{d \mathbf{P}_2}{dw} \mathbf{P}_2^{-1} \frac{d \mathbf{P}_2}{dw} + \mathbf{P}_2^{-1} \frac{d^2 \mathbf{P}_2}{dw^2}  \nonumber \\
&\quad +(\mathbf{P}_1+\mathbf{P}_2)^{-1} \frac{d (\mathbf{P}_1+\mathbf{P}_2)}{dw} (\mathbf{P}_1+\mathbf{P}_2)^{-1} \frac{d (\mathbf{P}_1+\mathbf{P}_2)}{dw} - (\mathbf{P}_1+\mathbf{P}_2)^{-1} \frac{d^2 (\mathbf{P}_1+\mathbf{P}_2)}{dw^2} \},
\end{align}

\begin{align} \label{eq:lndetP_3rdDif}
\frac{d^3}{dw^3} \ln \det \mathbf{P} = tr\{ 2 \mathbf{P}_1^{-1} \frac{d \mathbf{P}_1}{dw} \mathbf{P}_1^{-1} \frac{d \mathbf{P}_1}{dw} \mathbf{P}_1^{-1} \frac{d \mathbf{P}_1}{dw} - 3 \mathbf{P}_1^{-1} \frac{d \mathbf{P}_1}{dw} \mathbf{P}_1^{-1} \frac{d^2 \mathbf{P}_1}{d w^2} + \mathbf{P}_1^{-1} \frac{d^3 \mathbf{P}_1}{d w^3} & \nonumber \\
  + 2 \mathbf{P}_2^{-1} \frac{d \mathbf{P}_2}{dw} \mathbf{P}_2^{-1} \frac{d \mathbf{P}_2}{dw} \mathbf{P}_2^{-1} \frac{d \mathbf{P}_2}{dw} - 3 \mathbf{P}_2^{-1} \frac{d \mathbf{P}_2}{dw} \mathbf{P}_2^{-1} \frac{d^2 \mathbf{P}_2}{d w^2} + \mathbf{P}_2^{-1} \frac{d^3 \mathbf{P}_2}{d w^3} & \nonumber \\
  - 2 (\mathbf{P}_1+\mathbf{P}_2)^{-1} \frac{d (\mathbf{P}_1+\mathbf{P}_2)}{dw} (\mathbf{P}_1+\mathbf{P}_2)^{-1} \frac{d (\mathbf{P}_1+\mathbf{P}_2)}{dw} (\mathbf{P}_1+\mathbf{P}_2)^{-1} \frac{d (\mathbf{P}_1+\mathbf{P}_2)}{dw} & \nonumber \\
  + 3 (\mathbf{P}_1+\mathbf{P}_2)^{-1} \frac{d (\mathbf{P}_1+\mathbf{P}_2)}{dw} (\mathbf{P}_1+\mathbf{P}_2)^{-1} \frac{d^2 (\mathbf{P}_1+\mathbf{P}_2)}{d w^2} - (\mathbf{P}_1+\mathbf{P}_2)^{-1} \frac{d^3 (\mathbf{P}_1+\mathbf{P}_2)}{d w^3} & \},
\end{align}

\begin{align} \label{eq:lndetP_4thDif}
& \quad \frac{d^4}{dw^4} \ln \det \mathbf{P} = tr\{  \nonumber \\
& - 6 (\mathbf{P}_1^{-1} \frac{d \mathbf{P}_1}{dw})^4 + 12 (\mathbf{P}_1^{-1} \frac{d \mathbf{P}_1}{dw})^2 \mathbf{P}_1^{-1} \frac{d^2 \mathbf{P}_1}{d w^2} - 3 (\mathbf{P}_1^{-1} \frac{d^2 \mathbf{P}_1}{d w^2})^2 - 4 \mathbf{P}_1^{-1} \frac{d \mathbf{P}_1}{dw} \mathbf{P}_1^{-1} \frac{d^3 \mathbf{P}_1}{d w^3} + \mathbf{P}_1^{-1} \frac{d^4 \mathbf{P}_1}{d w^4}  \nonumber \\
& - 6 (\mathbf{P}_2^{-1} \frac{d \mathbf{P}_2}{dw})^4 + 12 (\mathbf{P}_2^{-1} \frac{d \mathbf{P}_2}{dw})^2 \mathbf{P}_2^{-1} \frac{d^2 \mathbf{P}_2}{d w^2} - 3 (\mathbf{P}_2^{-1} \frac{d^2 \mathbf{P}_2}{d w^2})^2 - 4 \mathbf{P}_2^{-1} \frac{d \mathbf{P}_2}{dw} \mathbf{P}_2^{-1} \frac{d^3 \mathbf{P}_2}{d w^3} + \mathbf{P}_2^{-1} \frac{d^4 \mathbf{P}_2}{d w^4}  \nonumber \\
& + 6 [(\mathbf{P}_1+\mathbf{P}_2)^{-1} \frac{d (\mathbf{P}_1+\mathbf{P}_2)}{dw}]^4 - 12 [(\mathbf{P}_1+\mathbf{P}_2)^{-1} \frac{d (\mathbf{P}_1+\mathbf{P}_2)}{dw}]^2 (\mathbf{P}_1+\mathbf{P}_2)^{-1} \frac{d^2 (\mathbf{P}_1+\mathbf{P}_2)}{d w^2}  \nonumber \\
& + 3 [(\mathbf{P}_1+\mathbf{P}_2)^{-1} \frac{d^2 (\mathbf{P}_1+\mathbf{P}_2)}{d w^2}]^2 + 4 (\mathbf{P}_1+\mathbf{P}_2)^{-1} \frac{d (\mathbf{P}_1+\mathbf{P}_2)}{dw} (\mathbf{P}_1+\mathbf{P}_2)^{-1} \frac{d^3 (\mathbf{P}_1+\mathbf{P}_2)}{d w^3}  \nonumber \\
& - (\mathbf{P}_1+\mathbf{P}_2)^{-1} \frac{d^4 (\mathbf{P}_1+\mathbf{P}_2)}{d w^4}\}.
\end{align}

\subsection{Variable transformation}

For $w \in (0,1)$, define the following notations
\begin{equation}  \label{eq:var_simple}
\mathbf{D}_1 (w) \equiv \frac{\mathbf{P}_{1d}}{w}, \qquad \mathbf{D}_2 (w) \equiv \frac{\mathbf{P}_{2d}}{1-w}, \qquad w_1 \equiv -w, \qquad w_2 \equiv 1-w.
\end{equation}
As $\mathbf{P}_{1d} \geq 0$ and $\mathbf{P}_{2d} \geq 0$, we also have $\mathbf{D}_1 \geq 0$, $\mathbf{D}_2 \geq 0$. Like $\mathbf{P}_{1d}$ and $\mathbf{P}_{2d}$, $\mathbf{D}_1$ and $\mathbf{D}_2$ are also symmetric matrices. From definitions given in (\ref{eq:defineP}) we have
\begin{subequations}  \label{eq:P1+P2_difs}
\begin{align}
\frac{d \mathbf{P}_1}{dw} &= \frac{\mathbf{D}_1}{w_1}, & \frac{d \mathbf{P}_2}{dw} &= \frac{\mathbf{D}_2}{w_2},  \\ 
\frac{d^2 \mathbf{P}_1}{d w^2} &= \frac{2 \mathbf{D}_1}{w_1^2}, & \frac{d^2 \mathbf{P}_2}{d w^2} &= \frac{2 \mathbf{D}_2}{w_2^2},  \\
\frac{d^3 \mathbf{P}_1}{d w^3} &= \frac{6 \mathbf{D}_1}{w_1^3}, & \frac{d^3 \mathbf{P}_2}{d w^3} &= \frac{6 \mathbf{D}_2}{w_2^3},  \\
\frac{d^4 \mathbf{P}_1}{d w^4} &= \frac{24 \mathbf{D}_1}{w_1^4}, & \frac{d^4 \mathbf{P}_2}{d w^4} &= \frac{24 \mathbf{D}_2}{w_2^4}.
\end{align}
\end{subequations}

Let $\mathbf{P}_3 = \mathbf{P}_1 + \mathbf{P}_2$. Substitute (\ref{eq:P1+P2_difs}) into (\ref{eq:lndetP_4thDif}) and obtain
\begin{align*}
& \quad \frac{d^4}{dw^4} \ln \det \mathbf{P} = tr\{  \nonumber \\
& - 6 (\mathbf{P}_1^{-1} \frac{\mathbf{D}_1}{w_1})^4 + 12 (\mathbf{P}_1^{-1} \frac{\mathbf{D}_1}{w_1})^2 \mathbf{P}_1^{-1} \frac{2 \mathbf{D}_1}{w_1^2} - 3 (\mathbf{P}_1^{-1} \frac{2 \mathbf{D}_1}{w_1^2})^2 - 4 \mathbf{P}_1^{-1} \frac{\mathbf{D}_1}{w_1} \mathbf{P}_1^{-1} \frac{6 \mathbf{D}_1}{w_1^3} + \mathbf{P}_1^{-1} \frac{24 \mathbf{D}_1}{w_1^4}  \nonumber \\
& - 6 (\mathbf{P}_2^{-1} \frac{\mathbf{D}_2}{w_2})^4 + 12 (\mathbf{P}_2^{-1} \frac{\mathbf{D}_2}{w_2})^2 \mathbf{P}_2^{-1} \frac{2 \mathbf{D}_2}{w_2^2} - 3 (\mathbf{P}_2^{-1} \frac{2 \mathbf{D}_2}{w_2^2})^2 - 4 \mathbf{P}_2^{-1} \frac{\mathbf{D}_2}{w_2} \mathbf{P}_2^{-1} \frac{6 \mathbf{D}_2}{w_2^3} + \mathbf{P}_2^{-1} \frac{24 \mathbf{D}_2}{w_2^4}  \nonumber \\
& + 6 [\mathbf{P}_3^{-1} (\frac{\mathbf{D}_1}{w_1} + \frac{\mathbf{D}_2}{w_2})]^4 - 12 [\mathbf{P}_3^{-1} (\frac{\mathbf{D}_1}{w_1} + \frac{\mathbf{D}_2}{w_2})]^2 \mathbf{P}_3^{-1} (\frac{2 \mathbf{D}_1}{w_1^2} + \frac{2 \mathbf{D}_2}{w_2^2})  \nonumber \\
& + 3 [\mathbf{P}_3^{-1} (\frac{2 \mathbf{D}_1}{w_1^2} + \frac{2 \mathbf{D}_2}{w_2^2})]^2 + 4 \mathbf{P}_3^{-1} (\frac{\mathbf{D}_1}{w_1} + \frac{\mathbf{D}_2}{w_2}) \mathbf{P}_3^{-1} (\frac{6 \mathbf{D}_1}{w_1^3} + \frac{6 \mathbf{D}_2}{w_2^3}) - \mathbf{P}_3^{-1} (\frac{24 \mathbf{D}_1}{w_1^4} + \frac{24 \mathbf{D}_2}{w_2^4})\}.
\end{align*}
Further expand above expression, collect relevant terms and derive
\begin{align}   \label{eq:lndetP_4thDif_2}
&\frac{d^4}{dw^4} \ln \det \mathbf{P} = 6 \cdot tr\{  \nonumber \\
& \quad \frac{1}{w_1^4} [(\mathbf{I} - \mathbf{P}_3^{-1} \mathbf{D}_1)^4 - (\mathbf{I} - \mathbf{P}_1^{-1} \mathbf{D}_1)^4] 
  + \frac{1}{w_2^4} [(\mathbf{I} - \mathbf{P}_3^{-1} \mathbf{D}_2)^4 - (\mathbf{I} - \mathbf{P}_2^{-1} \mathbf{D}_2)^4]  \nonumber \\
& + \frac{4}{w_1^3 w_2} \mathbf{P}_3^{-1} \mathbf{D}_2 \mathbf{P}_3^{-1} \mathbf{D}_1 (\mathbf{I} - \mathbf{P}_3^{-1} \mathbf{D}_1)^2
  + \frac{4}{w_1 w_2^3} \mathbf{P}_3^{-1} \mathbf{D}_1 \mathbf{P}_3^{-1} \mathbf{D}_2 (\mathbf{I} - \mathbf{P}_3^{-1} \mathbf{D}_2)^2 \nonumber \\
& + \frac{4}{w_1^2 w_2^2} \mathbf{P}_3^{-1} \mathbf{D}_1 (\mathbf{I} - \mathbf{P}_3^{-1} \mathbf{D}_1) \mathbf{P}_3^{-1} \mathbf{D}_2 (\mathbf{I} - \mathbf{P}_3^{-1} \mathbf{D}_2)
  + \frac{2}{w_1^2 w_2^2} (\mathbf{P}_3^{-1} \mathbf{D}_1 \mathbf{P}_3^{-1} \mathbf{D}_2)^2
\}
\end{align}

Further define the following notations
\begin{subequations}  \label{eq:vars_Q=Pinv}
\begin{align}
\mathbf{Q}_1 & \equiv \mathbf{P}_1^{-1}, & \mathbf{Q}_2 & \equiv \mathbf{P}_2^{-1}, & \mathbf{Q}_3 \equiv \mathbf{P}_1^{-1} + \mathbf{P}_2^{-1},  \\
\mathbf{R}_1 & \equiv \mathbf{Q}_1 \mathbf{Q}_3^{-1} \mathbf{Q}_1, & \mathbf{R}_2 & \equiv \mathbf{Q}_2 \mathbf{Q}_3^{-1} \mathbf{Q}_2. &
\end{align}
\end{subequations}
Then we have
\begin{align*}
\mathbf{P}_3^{-1} = (\mathbf{P}_1 + \mathbf{P}_2)^{-1} = \mathbf{P}_1^{-1} (\mathbf{P}_1^{-1} + \mathbf{P}_2^{-1})^{-1} \mathbf{P}_2^{-1} = \mathbf{P}_2^{-1} (\mathbf{P}_1^{-1} + \mathbf{P}_2^{-1})^{-1} \mathbf{P}_1^{-1}
\end{align*}
and
\begin{subequations}  \label{eq:vars_P3related}
\begin{align}
\mathbf{P}_3^{-1} & = \mathbf{Q}_1 \mathbf{Q}_3^{-1} \mathbf{Q}_2 = \mathbf{Q}_2 \mathbf{Q}_3^{-1} \mathbf{Q}_1 = \mathbf{Q}_1 - \mathbf{Q}_1 \mathbf{Q}_3^{-1} \mathbf{Q}_1 = \mathbf{Q}_2 - \mathbf{Q}_2 \mathbf{Q}_3^{-1} \mathbf{Q}_2, \\
\mathbf{I} - \mathbf{P}_3^{-1} \mathbf{D}_1 &= (\mathbf{I} - \mathbf{P}_1^{-1} \mathbf{D}_1) + (\mathbf{P}_1^{-1} \mathbf{D}_1 - \mathbf{P}_3^{-1} \mathbf{D}_1) = \mathbf{I} - \mathbf{Q}_1 \mathbf{D}_1 + \mathbf{R}_1 \mathbf{D}_1,  \\
\mathbf{I} - \mathbf{P}_3^{-1} \mathbf{D}_2 &= (\mathbf{I} - \mathbf{P}_2^{-1} \mathbf{D}_2) + (\mathbf{P}_2^{-1} \mathbf{D}_2 - \mathbf{P}_3^{-1} \mathbf{D}_2) = \mathbf{I} - \mathbf{Q}_2 \mathbf{D}_2 + \mathbf{R}_2 \mathbf{D}_2.
\end{align}
\end{subequations}

\subsection{Decomposition of the fourth-order derivative expression}

Use \textbf{Lemma} \ref{lm:trcyclic} (the cyclic property of trace operation) when necessary in following derivation and throughout this paper. Substitute (\ref{eq:vars_P3related}) into (\ref{eq:lndetP_4thDif_2}) and obtain
\begin{align}   \label{eq:lndetP_4thDif_decompose}
&\frac{d^4}{dw^4} \ln \det \mathbf{P} = 6 \cdot tr\{  \nonumber \\
& \quad \frac{1}{w_1^4} [(\mathbf{I} - \mathbf{Q}_1 \mathbf{D}_1 + \mathbf{R}_1 \mathbf{D}_1)^4 - (\mathbf{I} - \mathbf{Q}_1 \mathbf{D}_1)^4] 
  + \frac{1}{w_2^4} [(\mathbf{I} - \mathbf{Q}_2 \mathbf{D}_2 + \mathbf{R}_2 \mathbf{D}_2)^4 - (\mathbf{I} - \mathbf{Q}_2 \mathbf{D}_2)^4]  \nonumber \\
& + \frac{4}{w_1^3 w_2} \mathbf{P}_3^{-1} \mathbf{D}_2 \mathbf{P}_3^{-1} \mathbf{D}_1 (\mathbf{I} - \mathbf{Q}_1 \mathbf{D}_1 + \mathbf{R}_1 \mathbf{D}_1)^2
  + \frac{4}{w_1 w_2^3} \mathbf{P}_3^{-1} \mathbf{D}_1 \mathbf{P}_3^{-1} \mathbf{D}_2 (\mathbf{I} - \mathbf{Q}_2 \mathbf{D}_2 + \mathbf{R}_2 \mathbf{D}_2)^2 \nonumber \\
& + \frac{4}{w_1^2 w_2^2} \mathbf{P}_3^{-1} \mathbf{D}_1 (\mathbf{I} - \mathbf{Q}_1 \mathbf{D}_1 + \mathbf{R}_1 \mathbf{D}_1) \mathbf{P}_3^{-1} \mathbf{D}_2 (\mathbf{I} - \mathbf{Q}_2 \mathbf{D}_2 + \mathbf{R}_2 \mathbf{D}_2)
  + \frac{2}{w_1^2 w_2^2} (\mathbf{P}_3^{-1} \mathbf{D}_1 \mathbf{P}_3^{-1} \mathbf{D}_2)^2
\}  \nonumber \\
& = 6 \cdot F_1 + 24 \cdot F_2 + 24 \cdot F_3 + 12 \cdot F_{41} + 12 \cdot F_{42} + 24 \cdot F_{43},
\end{align}
where
\begin{align*}
F_1 &= tr\{ \frac{1}{w_1^4} (\mathbf{R}_1 \mathbf{D}_1)^4 + \frac{1}{w_2^4} (\mathbf{R}_2 \mathbf{D}_2)^4  \\
    & + \frac{4}{w_1^3 w_2} \mathbf{P}_3^{-1} \mathbf{D}_2 \mathbf{P}_3^{-1} \mathbf{D}_1 (\mathbf{R}_1 \mathbf{D}_1)^2 + \frac{4}{w_1 w_2^3} \mathbf{P}_3^{-1} \mathbf{D}_1 \mathbf{P}_3^{-1} \mathbf{D}_2 (\mathbf{R}_2 \mathbf{D}_2)^2  \\
    & + \frac{4}{w_1^2 w_2^2} \mathbf{P}_3^{-1} \mathbf{D}_1 \mathbf{R}_1 \mathbf{D}_1 \mathbf{P}_3^{-1} \mathbf{D}_2 \mathbf{R}_2 \mathbf{D}_2 + \frac{2}{w_1^2 w_2^2} (\mathbf{P}_3^{-1} \mathbf{D}_1 \mathbf{P}_3^{-1} \mathbf{D}_2)^2  \},  \\
F_2 &= tr\{ \frac{1}{w_1^4} [ (\mathbf{R}_1 \mathbf{D}_1)^3 (\mathbf{I} - \mathbf{Q}_1 \mathbf{D}_1) + (\mathbf{R}_1 \mathbf{D}_1)^2 (\mathbf{I} - \mathbf{Q}_1 \mathbf{D}_1)^2 + \mathbf{R}_1 \mathbf{D}_1 (\mathbf{I} - \mathbf{Q}_1 \mathbf{D}_1)^3 ]  \\
    & + \frac{1}{w_1^3 w_2} [ \mathbf{P}_3^{-1} \mathbf{D}_2 \mathbf{P}_3^{-1} \mathbf{D}_1 (\mathbf{I} - \mathbf{Q}_1 \mathbf{D}_1)^2 + \mathbf{R}_1 \mathbf{D}_1 \mathbf{P}_3^{-1} \mathbf{D}_2 \mathbf{P}_3^{-1} \mathbf{D}_1 (\mathbf{I} - \mathbf{Q}_1 \mathbf{D}_1)  \\
    &+ \mathbf{P}_3^{-1} \mathbf{D}_2 \mathbf{P}_3^{-1} \mathbf{D}_1 \mathbf{R}_1 \mathbf{D}_1 (\mathbf{I} - \mathbf{Q}_1 \mathbf{D}_1) ] + \frac{1}{w_1^2 w_2^2} \mathbf{P}_3^{-1} \mathbf{D}_2 \mathbf{R}_2 \mathbf{D}_2 \mathbf{P}_3^{-1} \mathbf{D}_1 (\mathbf{I} - \mathbf{Q}_1 \mathbf{D}_1)  \},  \\
F_3 &= tr\{ \frac{1}{w_2^4} [ (\mathbf{R}_2 \mathbf{D}_2)^3 (\mathbf{I} - \mathbf{Q}_2 \mathbf{D}_2) + (\mathbf{R}_2 \mathbf{D}_2)^2 (\mathbf{I} - \mathbf{Q}_2 \mathbf{D}_2)^2 + \mathbf{R}_2 \mathbf{D}_2 (\mathbf{I} - \mathbf{Q}_2 \mathbf{D}_2)^3 ]  \\
    & + \frac{1}{w_1 w_2^3} [ \mathbf{P}_3^{-1} \mathbf{D}_1 \mathbf{P}_3^{-1} \mathbf{D}_2 (\mathbf{I} - \mathbf{Q}_2 \mathbf{D}_2)^2 + \mathbf{R}_2 \mathbf{D}_2 \mathbf{P}_3^{-1} \mathbf{D}_1 \mathbf{P}_3^{-1} \mathbf{D}_2 (\mathbf{I} - \mathbf{Q}_2 \mathbf{D}_2)  \\
    &+ \mathbf{P}_3^{-1} \mathbf{D}_1 \mathbf{P}_3^{-1} \mathbf{D}_2 \mathbf{R}_2 \mathbf{D}_2 (\mathbf{I} - \mathbf{Q}_2 \mathbf{D}_2) ] + \frac{1}{w_1^2 w_2^2} \mathbf{P}_3^{-1} \mathbf{D}_1 \mathbf{R}_1 \mathbf{D}_1 \mathbf{P}_3^{-1} \mathbf{D}_2 (\mathbf{I} - \mathbf{Q}_2 \mathbf{D}_2)  \},  \\
F_{41} &= tr\{ \frac{1}{w_1^4} [ \mathbf{R}_1 \mathbf{D}_1 (\mathbf{I} - \mathbf{Q}_1 \mathbf{D}_1) ]^2 \},  \\
F_{42} &= tr\{ \frac{1}{w_2^4} [ \mathbf{R}_2 \mathbf{D}_2 (\mathbf{I} - \mathbf{Q}_2 \mathbf{D}_2) ]^2 \},  \\
F_{43} &= tr\{ \frac{1}{w_1^2 w_2^2} \mathbf{P}_3^{-1} \mathbf{D}_1 (\mathbf{I} - \mathbf{Q}_1 \mathbf{D}_1) \mathbf{P}_3^{-1} \mathbf{D}_2 (\mathbf{I} - \mathbf{Q}_2 \mathbf{D}_2) \}.
\end{align*}
In other words, the fourth-order derivative of the objective function $\ln \det(\mathbf{P}(w))$ is decomposed into six terms, namely the ones associated with $F_1$, $F_2$, $F_3$, $F_{41}$, $F_{42}$, and $F_{43}$ respectively. If each of the six terms is positive semi-definite, then the fourth-order derivative is naturally positive semi-definite.

In fact, we indeed have 
\begin{equation}  \label{eq:dif4_F1}
F_1 \geq 0
\end{equation}
the proof of which is provided in Appendix \ref{app:dif4_F1}, have
\begin{equation}  \label{eq:dif4_F2+F3}
F_2 \geq 0, \qquad F_3 \geq 0
\end{equation}
the proof of which is provided in Appendix \ref{app:dif4_F2+F3}, have
\begin{equation}  \label{eq:dif4_F41+F42}
F_{41} \geq 0, \qquad F_{42} \geq 0
\end{equation}
the proof of which is provided in Appendix \ref{app:dif4_F41+F42}, and have
\begin{equation}  \label{eq:dif4_F43}
F_{43} \geq 0
\end{equation}
the proof of which is provided in Appendix \ref{app:dif4_F43}. So the inequality (\ref{eq:convexC4}) is proved. In other words, the fourth-order convexity of the $w$-optimization problem is proved.

\section{Nested Newton method}  \label{sec:nested_newton}

For an objective function possessing the second-order convexity, one can apply the golden section method which has guaranteed performance. More specifically, the golden section method is guaranteed to converge at a linear-order rate. It is known that the Newton method is much faster, which can converge (if being able to converge) at a second-order rate \cite{Bonnans2006}. However, the Newton method does not have guaranteed performance and it may diverge.

The Newton method is normally applied when the initial estimate is guessed to be ``close enough'' to the optimum. However, what is the mathematical criterion that defines ``being close enough'' rigorously? It is indeed difficult to give an answer (especially a general answer) to the question. In practical applications, \textit{ad hoc} heuristics and assumptions tend to be involved, which might be a kind of empirically working expedient but is at least theoretically unsound.

Since the Newton method has the merit of being fast but has the drawback of unguaranteed performance, the author ventures a question: Can we have a variant of the Newton method that is guaranteed to converge? In other words, is there a way to keep the merit of the Newton method on one hand and to remove its drawback on the other hand? 

If the objective function further possesses the fourth-order convexity, then the answer to above ventured question is fortunately \textit{yes}. This section presents such a method, which is coined as the \textbf{nested Newton method} --- The term ``nested'' here is rather in the sense as in the \textit{nested interval theorem}, instead of as in some nested control system architectures.

\subsection{Formalism of nested iteration}

Given a generic objective function $f(x)$ that possesses both the second-order convexity and the fourth-order convexity, namely a $f(x)$ satisfies both
\begin{equation} \label{eq:fx_convexC2}
\frac{d^2}{d x^2} f(x) \equiv \ddot{f}(x) \geq 0
\end{equation}
and
\begin{equation} \label{eq:fx_convexC4}
\frac{d^4}{d x^4} f(x) = \frac{d^2}{d x^2} \ddot{f}(x) \geq 0.
\end{equation}
The inequality (\ref{eq:fx_convexC2}) implies that the function $f(x)$ itself is convex (or figuratively speaking, bending downward) whereas the inequality (\ref{eq:fx_convexC4}) implies that the second-order derivative of $f(x)$, i.e., $\ddot{f}(x)$, is also convex.

Finding the optimum of $f(x)$ is equivalent to finding the root $x^*$ of $\dot{f}(x) = 0$. Given two bounding points $x_a$ and $x_b$ that satisfy $\dot{f}(x_a) < 0$ and $\dot{f}(x_b)>0$ and hence contain the root $x^*$ inside the interval $(x_a, x_b)$, the nested Newton method consists in iteration from both ends of the interval $(x_a, x_b)$ to the two ends of a (much) smaller interval $(x_a^*, x_b^*)$ as
\begin{subequations}  \label{eq:nested_newton}
\begin{align}
x_a^* &= x_a - \frac{2 \dot{f}(x_a)}{\ddot{f}(x_a) + \sqrt{\ddot{f}(x_a)^2 - 2 \dot{f}(x_a) k}},  \\
x_b^* &= x_b - \frac{2 \dot{f}(x_b)}{\ddot{f}(x_b) + \sqrt{\ddot{f}(x_b)^2 - 2 \dot{f}(x_b) k}},
\end{align}
\end{subequations}
where
\begin{align*}
k \equiv \frac{\ddot{f}(x_b) - \ddot{f}(x_a)}{x_b - x_a}.
\end{align*}
If the adaptive factor $k$ in (\ref{eq:nested_newton}) is fixed to zero, then (\ref{eq:nested_newton}) is reduced to the Newton method. So (\ref{eq:nested_newton}) can be regarded as a variant of the Newton method and also enjoys the merit of being fast.

\subsection{Proof of the nestedness (and hence convergence)}  \label{sec:proof_nestedness}

This subsection proves the nestedness of the iteration (\ref{eq:nested_newton}), namely
\begin{equation}  \label{eq:nestedness}
x^* \in [x_a^*, x_b^*] \subset [x_a, x_b]. 
\end{equation}

Consider $x_a$ and $x_a^*$. According to (\ref{eq:fx_convexC2}) and the assumed condition $\dot{f}(x_a) < 0$, it is apparent that $x_a < x_a^*$. So the key point is to prove
\begin{equation}  \label{eq:xa_star_leq_x_star}
x_a^* \leq x^*.
\end{equation}
Construct a function $g(x)$ such that
\begin{equation}  \label{eq:gx_def}
g(x) = \dot{f}(x_a) + \ddot{f}(x_a)(x - x_a) + \frac{k}{2} (x - x_a)^2,
\end{equation}
where $k$ is the same to that specified in (\ref{eq:nested_newton}) and $x \in \Omega_x \equiv \{x | x \geq x_a , g(x) \leq 0\}$.

Note that $\ddot{f}(x)$ is convex according to (\ref{eq:fx_convexC4}) and $x^* \in (x_a, x_b)$, so we have
\begin{align*}
\dot{g}(x) = \ddot{f}(x_a) + k (x - x_a) = \ddot{f}(x_a) + \frac{\ddot{f}(x_b) - \ddot{f}(x_a)}{x_b - x_a} (x - x_a) \geq \ddot{f}(x).
\end{align*}
Also note that $g(x_a) = \dot{f}(x_a)$, so
\begin{equation}  \label{eq:gx_geq_dot_f}
g(x) \geq \dot{f}(x)
\end{equation}
for $x \in \Omega_x$. Since $\dot{f}(x)$ increases monotonously, $g(x)$ increases also monotonously and above $\dot{f}(x)$, until it arrives at the right end of $\Omega_x$, namely the point $x$ at which $g(x)$ touches the horizontal line $g(x) = 0$, which is right $x_a^*$ given in (\ref{eq:nested_newton}) --- The trivial verification of why $x_a^*$ is the right end of $\Omega_x$ in the three cases $k<0$, $k=0$, and $k>0$ respectively is left to readers. 

Combining the fact $\dot{f}(x)$ increases monotonously, the fact $\dot{f}(x^*) = 0$, the fact (\ref{eq:gx_geq_dot_f}), and the fact $g(x_a^*) = 0$, we can easily conclude (\ref{eq:xa_star_leq_x_star}). Similarly for $x_b$ and $x_b^*$ we can also prove
\begin{align*}
x^* \leq x_b^* < x_b
\end{align*}
and (\ref{eq:nestedness}) holds. So the nested Newton method formalized in (\ref{eq:nested_newton}) is guaranteed to converge.

\subsection{Discussion}

Substitute the objective function $\ln \det(\mathbf{P}(w))$ in the $w$-optimization problem for the generic objective function $f(x)$ above and obtain a guaranteed fast implementation of the Split CIF.

In practical applications, the golden section method and the nested Newton method, both of which have guaranteed performance, can be used together. The golden section method is applied for only few iterations (say three iterations) at the very beginning to provide a confined interval $[w_a, w_b]$ that contains the optimal $w$. Then from the two bounding points $w_a$ and $w_b$, the nested Newton method is applied iteratively
\footnote{In a similar spirit, the golden section method and the original Newton method can be used together as well. However, the original Newton method has no guaranteed performance and it is difficult to prescribe a proper number of iterations for the golden section method which is applied first before the Newton method --- The argument here is corresponding to the comments at the beginning of Section \ref{sec:nested_newton} --- Besides, it is hard to have a mathematically rigorous way of estimating the optimization accuracy without resorting to any \textit{ad hoc} heuristic or assumption.}. 
According to the author's experience, only one or two further iterations via (\ref{eq:nested_newton}) would often converge to very accurate result of $w$.

Another merit of the nested Newton method, which is a ``by-product'' of the proof provided in Section \ref{sec:proof_nestedness}, is that we can clearly know the optimization accuracy after each iteration via (\ref{eq:nested_newton}). This gives us a natural way of defining the termination criterion.

It is worth noting that rigorous mathematical details for many trivial considerations are intentionally saved in this paper, yet the author believes readers would complement all necessary trivial details by themselves. 

\section{Conclusion}

The $w$-optimization problem of the split covariance intersection filter (or Split CIF for short) is revisited. It enjoys not only the desirable property of convexity (or more clearly, the second-order convexity in this paper's context) but also a more desirable property namely the fourth-order convexity. This paper proves the fourth-order convexity of the $w$-optimization problem, based on which this paper also presents a guaranteed fast implementation of the Split CIF, namely the nested Newton method.

Although the author has explained in Section \ref{sec:woptprob} why the determinant function $\det(\cdot)$ is preferred over the trace function $tr \{ \cdot \}$ for the $w$-optimization problem in practical applications, it can still be examined theoretically whether the $w$-optimization problem also possesses the fourth-order convexity if the trace function $tr \{ \cdot \}$ is adopted.

Besides, more advanced numeric optimization methods that may better take advantage of a generic objective function possessing both the second-order convexity and the fourth-order convexity are worth further study.

\appendix

\section{Positive semi-definiteness of the six terms}

\subsection{Positive semi-definiteness of $F_1$}  \label{app:dif4_F1}

Define the following notations
\begin{equation}  \label{eq:vars_G1+G2}
\mathbf{G}_1 \equiv \sqrt{\mathbf{Q}_3^{-1}} \mathbf{Q}_1 \mathbf{D}_1 \mathbf{Q}_1 \sqrt{\mathbf{Q}_3^{-1}}, \qquad \mathbf{G}_2 \equiv \sqrt{\mathbf{Q}_3^{-1}} \mathbf{Q}_2 \mathbf{D}_2 \mathbf{Q}_2 \sqrt{\mathbf{Q}_3^{-1}}.
\end{equation}
Since $\mathbf{Q}_1$, $\mathbf{Q}_2$, $\mathbf{Q}_3$, $\mathbf{D}_1$, and $\mathbf{D}_2$ are symmetric matrices (more specifically, positive semi-definite matrices), $\mathbf{G}_1$ and $\mathbf{G}_2$ are so as well.

Recall the second equation of (\ref{eq:vars_Q=Pinv}) namely
\begin{align*}
\mathbf{R}_1 \equiv \mathbf{Q}_1 \mathbf{Q}_3^{-1} \mathbf{Q}_1, \qquad \mathbf{R}_2 \equiv \mathbf{Q}_2 \mathbf{Q}_3^{-1} \mathbf{Q}_2
\end{align*}
and the first equation of (\ref{eq:vars_P3related}) namely
\begin{align*}
\mathbf{P}_3^{-1} & = \mathbf{Q}_1 \mathbf{Q}_3^{-1} \mathbf{Q}_2 = \mathbf{Q}_2 \mathbf{Q}_3^{-1} \mathbf{Q}_1.
\end{align*}
Then we have
\begin{align}  \label{eq:dif4_F1_a}
tr\{ (\mathbf{R}_1 \mathbf{D}_1)^4 \} &= tr\{ (\mathbf{Q}_1 \mathbf{Q}_3^{-1} \mathbf{Q}_1 \mathbf{D}_1)^4 \} = tr\{ \mathbf{Q}_1 \sqrt{\mathbf{Q}_3^{-1}} \sqrt{\mathbf{Q}_3^{-1}} \mathbf{Q}_1 \mathbf{D}_1 (\mathbf{Q}_1 \sqrt{\mathbf{Q}_3^{-1}} \sqrt{\mathbf{Q}_3^{-1}} \mathbf{Q}_1 \mathbf{D}_1)^3 \}  \nonumber \\
  &= tr\{ \sqrt{\mathbf{Q}_3^{-1}} \mathbf{Q}_1 \mathbf{D}_1 (\mathbf{Q}_1 \sqrt{\mathbf{Q}_3^{-1}} \sqrt{\mathbf{Q}_3^{-1}} \mathbf{Q}_1 \mathbf{D}_1)^3 \mathbf{Q}_1 \sqrt{\mathbf{Q}_3^{-1}} \} = tr\{ \mathbf{G}_1^4 \}.
\end{align}
Similarly we have
\begin{align}  \label{eq:dif4_F1_b}
tr\{ (\mathbf{R}_2 \mathbf{D}_2)^4 \} = tr\{ \mathbf{G}_2^4 \}.
\end{align}
We have
\begin{align}  \label{eq:dif4_F1_c}
& tr\{ \mathbf{P}_3^{-1} \mathbf{D}_2 \mathbf{P}_3^{-1} \mathbf{D}_1 (\mathbf{R}_1 \mathbf{D}_1)^2 \} = tr\{ \mathbf{Q}_1 \mathbf{Q}_3^{-1} \mathbf{Q}_2 \mathbf{D}_2 \mathbf{Q}_2 \mathbf{Q}_3^{-1} \mathbf{Q}_1 \mathbf{D}_1 (\mathbf{Q}_1 \mathbf{Q}_3^{-1} \mathbf{Q}_1 \mathbf{D}_1)^2 \}  \nonumber \\
& \qquad = tr\{ \mathbf{Q}_1 \sqrt{\mathbf{Q}_3^{-1}} \sqrt{\mathbf{Q}_3^{-1}} \mathbf{Q}_2 \mathbf{D}_2 \mathbf{Q}_2 \sqrt{\mathbf{Q}_3^{-1}} \sqrt{\mathbf{Q}_3^{-1}} \mathbf{Q}_1 \mathbf{D}_1 (\mathbf{Q}_1 \sqrt{\mathbf{Q}_3^{-1}} \sqrt{\mathbf{Q}_3^{-1}} \mathbf{Q}_1 \mathbf{D}_1)^2 \}  \nonumber \\
& \qquad = tr\{ \sqrt{\mathbf{Q}_3^{-1}} \mathbf{Q}_1 \mathbf{D}_1 (\mathbf{Q}_1 \sqrt{\mathbf{Q}_3^{-1}} \sqrt{\mathbf{Q}_3^{-1}} \mathbf{Q}_1 \mathbf{D}_1)^2 \mathbf{Q}_1 \sqrt{\mathbf{Q}_3^{-1}} \sqrt{\mathbf{Q}_3^{-1}} \mathbf{Q}_2 \mathbf{D}_2 \mathbf{Q}_2 \sqrt{\mathbf{Q}_3^{-1}} \}  \nonumber \\
& \qquad = tr\{ \mathbf{G}_1^3 \mathbf{G}_2 \}.
\end{align}
Similarly we have
\begin{align}  \label{eq:dif4_F1_d}
tr\{ \mathbf{P}_3^{-1} \mathbf{D}_1 \mathbf{P}_3^{-1} \mathbf{D}_2 (\mathbf{R}_2 \mathbf{D}_2)^2 \} = tr\{ \mathbf{G}_1 \mathbf{G}_2^3 \}.
\end{align}
We have
\begin{align}  \label{eq:dif4_F1_e}
& tr\{ \mathbf{P}_3^{-1} \mathbf{D}_1 \mathbf{R}_1 \mathbf{D}_1 \mathbf{P}_3^{-1} \mathbf{D}_2 \mathbf{R}_2 \mathbf{D}_2 \} = tr\{ \mathbf{Q}_2 \mathbf{Q}_3^{-1} \mathbf{Q}_1 \mathbf{D}_1 \mathbf{Q}_1 \mathbf{Q}_3^{-1} \mathbf{Q}_1 \mathbf{D}_1 \mathbf{Q}_1 \mathbf{Q}_3^{-1} \mathbf{Q}_2 \mathbf{D}_2 \mathbf{Q}_2 \mathbf{Q}_3^{-1} \mathbf{Q}_2 \mathbf{D}_2 \}  \nonumber \\
& = tr\{ \mathbf{Q}_2 \sqrt{\mathbf{Q}_3^{-1}} \sqrt{\mathbf{Q}_3^{-1}} \mathbf{Q}_1 \mathbf{D}_1 \mathbf{Q}_1 \sqrt{\mathbf{Q}_3^{-1}} \sqrt{\mathbf{Q}_3^{-1}} \mathbf{Q}_1 \mathbf{D}_1 \mathbf{Q}_1 \sqrt{\mathbf{Q}_3^{-1}} \sqrt{\mathbf{Q}_3^{-1}} \mathbf{Q}_2 \mathbf{D}_2 \mathbf{Q}_2 \sqrt{\mathbf{Q}_3^{-1}} \sqrt{\mathbf{Q}_3^{-1}} \mathbf{Q}_2 \mathbf{D}_2 \}  \nonumber \\
& = tr\{ \sqrt{\mathbf{Q}_3^{-1}} \mathbf{Q}_1 \mathbf{D}_1 \mathbf{Q}_1 \sqrt{\mathbf{Q}_3^{-1}} \sqrt{\mathbf{Q}_3^{-1}} \mathbf{Q}_1 \mathbf{D}_1 \mathbf{Q}_1 \sqrt{\mathbf{Q}_3^{-1}} \sqrt{\mathbf{Q}_3^{-1}} \mathbf{Q}_2 \mathbf{D}_2 \mathbf{Q}_2 \sqrt{\mathbf{Q}_3^{-1}} \sqrt{\mathbf{Q}_3^{-1}} \mathbf{Q}_2 \mathbf{D}_2 \mathbf{Q}_2 \sqrt{\mathbf{Q}_3^{-1}} \}  \nonumber \\
& = tr\{ \mathbf{G}_1^2 \mathbf{G}_2^2 \}
\end{align}
and
\begin{align}  \label{eq:dif4_F1_f}
& tr\{ (\mathbf{P}_3^{-1} \mathbf{D}_1 \mathbf{P}_3^{-1} \mathbf{D}_2)^2 \} = tr\{ (\mathbf{Q}_2 \mathbf{Q}_3^{-1} \mathbf{Q}_1 \mathbf{D}_1 \mathbf{Q}_1 \mathbf{Q}_3^{-1} \mathbf{Q}_2 \mathbf{D}_2)^2 \}  \nonumber \\
& = tr\{ \mathbf{Q}_2 \sqrt{\mathbf{Q}_3^{-1}} \sqrt{\mathbf{Q}_3^{-1}} \mathbf{Q}_1 \mathbf{D}_1 \mathbf{Q}_1 \sqrt{\mathbf{Q}_3^{-1}} \sqrt{\mathbf{Q}_3^{-1}} \mathbf{Q}_2 \mathbf{D}_2 \mathbf{Q}_2 \sqrt{\mathbf{Q}_3^{-1}} \sqrt{\mathbf{Q}_3^{-1}} \mathbf{Q}_1 \mathbf{D}_1 \mathbf{Q}_1 \sqrt{\mathbf{Q}_3^{-1}} \sqrt{\mathbf{Q}_3^{-1}} \mathbf{Q}_2 \mathbf{D}_2 \}  \nonumber \\
& = tr\{ \sqrt{\mathbf{Q}_3^{-1}} \mathbf{Q}_1 \mathbf{D}_1 \mathbf{Q}_1 \sqrt{\mathbf{Q}_3^{-1}} \sqrt{\mathbf{Q}_3^{-1}} \mathbf{Q}_2 \mathbf{D}_2 \mathbf{Q}_2 \sqrt{\mathbf{Q}_3^{-1}} \sqrt{\mathbf{Q}_3^{-1}} \mathbf{Q}_1 \mathbf{D}_1 \mathbf{Q}_1 \sqrt{\mathbf{Q}_3^{-1}} \sqrt{\mathbf{Q}_3^{-1}} \mathbf{Q}_2 \mathbf{D}_2 \mathbf{Q}_2 \sqrt{\mathbf{Q}_3^{-1}} \}  \nonumber \\
& = tr\{ (\mathbf{G}_1 \mathbf{G}_2)^2 \}.
\end{align}

Substitute (\ref{eq:dif4_F1_a}), (\ref{eq:dif4_F1_b}), (\ref{eq:dif4_F1_c}), (\ref{eq:dif4_F1_d}), (\ref{eq:dif4_F1_e}), and (\ref{eq:dif4_F1_f}) into $F_1$ of (\ref{eq:lndetP_4thDif_decompose}) and obtain
\begin{align*}
F_1 &= tr\{ \frac{1}{w_1^4} \mathbf{G}_1^4 + \frac{1}{w_2^4} \mathbf{G}_2^4 + \frac{4}{w_1^3 w_2} \mathbf{G}_1^3 \mathbf{G}_2 + \frac{4}{w_1 w_2^3} \mathbf{G}_1 \mathbf{G}_2^3 + \frac{4}{w_1^2 w_2^2} \mathbf{G}_1^2 \mathbf{G}_2^2 + \frac{2}{w_1^2 w_2^2} (\mathbf{G}_1 \mathbf{G}_2)^2 \}   \\
  &= tr\{ (\frac{\mathbf{G}_1}{w_1} + \frac{\mathbf{G}_2}{w_2})^4 \}.
\end{align*}
Since the matrix
\begin{align*}
\mathbf{G}_{12} \equiv \frac{\mathbf{G}_1}{w_1} + \frac{\mathbf{G}_2}{w_2}
\end{align*}
is symmetric, the matrix
\begin{align*}
\mathbf{G}_{12}^4 = (\mathbf{G}_{12} \mathbf{G}_{12})^{\mathrm{T}} (\mathbf{G}_{12} \mathbf{G}_{12})
\end{align*}
is positive semi-definite and
\begin{align*}
F_1 = tr\{ \mathbf{G}_{12}^4 \} \geq 0.
\end{align*}
The inequality (\ref{eq:dif4_F1}) is proved.

\subsection{Positive semi-definiteness of $F_2$ and $F_3$}  \label{app:dif4_F2+F3}

Besides the notations $\mathbf{G}_1$ and $\mathbf{G}_2$ defined in (\ref{eq:vars_G1+G2}), further define the following notations
\begin{equation}  \label{eq:vars_S1+E1}
\mathbf{S}_1 \equiv \sqrt{\mathbf{D}_1} \mathbf{Q}_1 \sqrt{\mathbf{Q}_3^{-1}}, \qquad \mathbf{E}_1 \equiv \mathbf{I} - \sqrt{\mathbf{D}_1} \mathbf{Q}_1 \sqrt{\mathbf{D}_1}.
\end{equation}
Since $\mathbf{Q}_1 \equiv \mathbf{P}_1^{-1}$ and $\mathbf{D}_1$ are symmetric matrices, $\mathbf{E}_1$ is so as well. Besides, note that
\begin{align*}
\mathbf{P}_1 \geq \mathbf{D}_1 \geq 0,
\end{align*}
then we have
\footnote{It is worth noting that $\mathbf{P}_1$ is always positive definite in practical applications and $\mathbf{Q}_1 \equiv \mathbf{P}_1^{-1}$ is well-defined mathematically, whereas $\mathbf{D}_1^{-1}$ is only positive semi-definite and $\mathbf{D}_1^{-1}$ may not exist in rigorous sense. However, the inequality (\ref{eq:D1inv_geq_Q1inv}) is still meaningful, which can be understood as follows: Since $\mathbf{P}_1 > 0$, we can always perturb $\mathbf{D}_1$ by an arbitrarily infinitesimal positive semi-definite matrix $\mathbf{M}_{\epsilon}$ such that
\begin{align*}
\mathbf{P}_1 \geq \mathbf{D}_1 + \mathbf{M}_{\epsilon} > 0 \iff (\mathbf{D}_1 + \mathbf{M}_{\epsilon})^{-1} \geq \mathbf{P}_1^{-1} > 0.
\end{align*}
Take the limit as $\mathbf{M}_{\epsilon} \to 0$ and hence obtain (\ref{eq:D1inv_geq_Q1inv}). In fact, for $\mathbf{D}_1^{-1}$ which may not exist in rigorous sense, we can still understand it in a meaningful logic somehow like that we have the value (positive-side) zero, i.e., $0_+$, and then we have $0_+^{-1} = +\infty$, i.e., an arbitrarily large value. In such way, (\ref{eq:D1inv_geq_Q1inv}) can even be directly understood, without resorting to the perturbation-limit skill.}
\begin{equation}  \label{eq:D1inv_geq_Q1inv}
\mathbf{D}_1^{-1} \geq \mathbf{P}_1^{-1} > 0  \implies \mathbf{D}_1^{-1} - \mathbf{P}_1^{-1} \geq 0
\end{equation}
and hence have
\begin{align}  \label{eq:vars_E1_pos}
\mathbf{E}_1 = \sqrt{\mathbf{D}_1} \mathbf{D}_1^{-1} \sqrt{\mathbf{D}_1} - \sqrt{\mathbf{D}_1} \mathbf{P}_1^{-1} \sqrt{\mathbf{D}_1} = \sqrt{\mathbf{D}_1} (\mathbf{D}_1^{-1} - \mathbf{P}_1^{-1}) \sqrt{\mathbf{D}_1} \geq 0.
\end{align}

Recall the second equation of (\ref{eq:vars_Q=Pinv}) and the first equation of (\ref{eq:vars_P3related}) again. Then for the term $(\mathbf{R}_1 \mathbf{D}_1)^3 (\mathbf{I} - \mathbf{Q}_1 \mathbf{D}_1)$ of $F_2$, we have
\begin{align*}
& tr\{ (\mathbf{R}_1 \mathbf{D}_1)^3 (\mathbf{I} - \mathbf{Q}_1 \mathbf{D}_1) \} = tr\{ (\mathbf{Q}_1 \mathbf{Q}_3^{-1} \mathbf{Q}_1 \mathbf{D}_1)^3 (\mathbf{I} - \mathbf{Q}_1 \mathbf{D}_1) \}  \\
& \qquad = tr\{ \mathbf{Q}_1 \sqrt{\mathbf{Q}_3^{-1}} (\sqrt{\mathbf{Q}_3^{-1}} \mathbf{Q}_1 \mathbf{D}_1 \mathbf{Q}_1 \sqrt{\mathbf{Q}_3^{-1}})^2 (\sqrt{\mathbf{Q}_3^{-1}} \mathbf{Q}_1 \sqrt{\mathbf{D}_1}) (\mathbf{I} - \sqrt{\mathbf{D}_1} \mathbf{Q}_1 \sqrt{\mathbf{D}_1}) \sqrt{\mathbf{D}_1} \}  \\
& \qquad = tr\{ (\sqrt{\mathbf{Q}_3^{-1}} \mathbf{Q}_1 \mathbf{D}_1 \mathbf{Q}_1 \sqrt{\mathbf{Q}_3^{-1}})^2 (\sqrt{\mathbf{Q}_3^{-1}} \mathbf{Q}_1 \sqrt{\mathbf{D}_1}) (\mathbf{I} - \sqrt{\mathbf{D}_1} \mathbf{Q}_1 \sqrt{\mathbf{D}_1}) \sqrt{\mathbf{D}_1} \mathbf{Q}_1 \sqrt{\mathbf{Q}_3^{-1}} \}  \\
& \qquad = tr\{ \mathbf{G}_1^2 \mathbf{S}_1^{\mathrm{T}} \mathbf{E}_1 \mathbf{S}_1 \}.
\end{align*}
Similar derivation applies to the terms $(\mathbf{R}_1 \mathbf{D}_1)^2 (\mathbf{I} - \mathbf{Q}_1 \mathbf{D}_1)^2$ and $\mathbf{R}_1 \mathbf{D}_1 (\mathbf{I} - \mathbf{Q}_1 \mathbf{D}_1)^3$ of $F_2$ and together we have
\begin{subequations}  \label{eq:dif4_F2_a}
\begin{align}  
tr\{ (\mathbf{R}_1 \mathbf{D}_1)^3 (\mathbf{I} - \mathbf{Q}_1 \mathbf{D}_1) \} &= tr\{ \mathbf{G}_1^2 \mathbf{S}_1^{\mathrm{T}} \mathbf{E}_1 \mathbf{S}_1 \},  \\
tr\{ (\mathbf{R}_1 \mathbf{D}_1)^2 (\mathbf{I} - \mathbf{Q}_1 \mathbf{D}_1)^2 \} &= tr\{ \mathbf{G}_1 \mathbf{S}_1^{\mathrm{T}} \mathbf{E}_1^2 \mathbf{S}_1 \},  \\
tr\{ \mathbf{R}_1 \mathbf{D}_1 (\mathbf{I} - \mathbf{Q}_1 \mathbf{D}_1)^3 \} &= tr\{ \mathbf{S}_1^{\mathrm{T}} \mathbf{E}_1^3 \mathbf{S}_1 \}.
\end{align}
\end{subequations}
We also have
\begin{align}  \label{eq:dif4_F2_b}
& tr\{ \mathbf{P}_3^{-1} \mathbf{D}_2 \mathbf{P}_3^{-1} \mathbf{D}_1 (\mathbf{I} - \mathbf{Q}_1 \mathbf{D}_1)^2 \} = tr\{ \mathbf{Q}_1 \mathbf{Q}_3^{-1} \mathbf{Q}_2 \mathbf{D}_2 \mathbf{Q}_2 \mathbf{Q}_3^{-1} \mathbf{Q}_1 \mathbf{D}_1 (\mathbf{I} - \mathbf{Q}_1 \mathbf{D}_1)^2 \}  \nonumber \\
& \qquad = tr\{ \mathbf{Q}_1 \sqrt{\mathbf{Q}_3^{-1}} (\sqrt{\mathbf{Q}_3^{-1}} \mathbf{Q}_2 \mathbf{D}_2 \mathbf{Q}_2 \sqrt{\mathbf{Q}_3^{-1}}) (\sqrt{\mathbf{Q}_3^{-1}} \mathbf{Q}_1 \sqrt{\mathbf{D}_1}) (\mathbf{I} - \sqrt{\mathbf{D}_1} \mathbf{Q}_1 \sqrt{\mathbf{D}_1})^2 \sqrt{\mathbf{D}_1} \}  \nonumber \\
& \qquad = tr\{ (\sqrt{\mathbf{Q}_3^{-1}} \mathbf{Q}_2 \mathbf{D}_2 \mathbf{Q}_2 \sqrt{\mathbf{Q}_3^{-1}}) (\sqrt{\mathbf{Q}_3^{-1}} \mathbf{Q}_1 \sqrt{\mathbf{D}_1}) (\mathbf{I} - \sqrt{\mathbf{D}_1} \mathbf{Q}_1 \sqrt{\mathbf{D}_1})^2 \sqrt{\mathbf{D}_1} \mathbf{Q}_1 \sqrt{\mathbf{Q}_3^{-1}} \}  \nonumber \\
& \qquad = tr\{ \mathbf{G}_2 \mathbf{S}_1^{\mathrm{T}} \mathbf{E}_1^2 \mathbf{S}_1 \},
\end{align}
and similarly have
\begin{subequations} \label{eq:dif4_F2_c}
\begin{align}
tr\{ \mathbf{R}_1 \mathbf{D}_1 \mathbf{P}_3^{-1} \mathbf{D}_2 \mathbf{P}_3^{-1} \mathbf{D}_1 (\mathbf{I} - \mathbf{Q}_1 \mathbf{D}_1) \} &= tr\{ \mathbf{G}_1 \mathbf{G}_2 \mathbf{S}_1^{\mathrm{T}} \mathbf{E}_1 \mathbf{S}_1 \},\\
tr\{ \mathbf{P}_3^{-1} \mathbf{D}_2 \mathbf{P}_3^{-1} \mathbf{D}_1 \mathbf{R}_1 \mathbf{D}_1 (\mathbf{I} - \mathbf{Q}_1 \mathbf{D}_1) \} &= tr\{ \mathbf{G}_2 \mathbf{G}_1 \mathbf{S}_1^{\mathrm{T}} \mathbf{E}_1 \mathbf{S}_1 \},\\
tr\{ \mathbf{P}_3^{-1} \mathbf{D}_2 \mathbf{R}_2 \mathbf{D}_2 \mathbf{P}_3^{-1} \mathbf{D}_1 (\mathbf{I} - \mathbf{Q}_1 \mathbf{D}_1) \} &= tr\{ \mathbf{G}_2^2 \mathbf{S}_1^{\mathrm{T}} \mathbf{E}_1 \mathbf{S}_1 \}.
\end{align}
\end{subequations}

Substitute (\ref{eq:dif4_F2_a}), (\ref{eq:dif4_F2_b}), and (\ref{eq:dif4_F2_c}) into $F_2$ of (\ref{eq:lndetP_4thDif_decompose}) and obtain
\begin{align*}
F_2 &= tr\{ \frac{1}{w_1^4} ( \mathbf{G}_1^2 \mathbf{S}_1^{\mathrm{T}} \mathbf{E}_1 \mathbf{S}_1 + \mathbf{G}_1 \mathbf{S}_1^{\mathrm{T}} \mathbf{E}_1^2 \mathbf{S}_1 + \mathbf{S}_1^{\mathrm{T}} \mathbf{E}_1^3 \mathbf{S}_1 )  \\
    & + \frac{1}{w_1^3 w_2} ( \mathbf{G}_2 \mathbf{S}_1^{\mathrm{T}} \mathbf{E}_1^2 \mathbf{S}_1 + \mathbf{G}_1 \mathbf{G}_2 \mathbf{S}_1^{\mathrm{T}} \mathbf{E}_1 \mathbf{S}_1 + \mathbf{G}_2 \mathbf{G}_1 \mathbf{S}_1^{\mathrm{T}} \mathbf{E}_1 \mathbf{S}_1 ) + \frac{1}{w_1^2 w_2^2} \mathbf{G}_2^2 \mathbf{S}_1^{\mathrm{T}} \mathbf{E}_1 \mathbf{S}_1  \}
\end{align*}
Recall (\ref{eq:var_simple}) and note that
\begin{align*}
| w_1 | = - w_1 = w > 0, 
\end{align*}
so we have
\begin{align}  \label{eq:dif4_F2_new}
F_2 &= tr\{ \frac{1}{| w_1 |^4} ( \mathbf{G}_1^2 \mathbf{S}_1^{\mathrm{T}} \mathbf{E}_1 \mathbf{S}_1 + \mathbf{G}_1 \mathbf{S}_1^{\mathrm{T}} \mathbf{E}_1^2 \mathbf{S}_1 + \mathbf{S}_1^{\mathrm{T}} \mathbf{E}_1^3 \mathbf{S}_1 )  \nonumber \\
    & - \frac{1}{| w_1 |^3 w_2} ( \mathbf{G}_2 \mathbf{S}_1^{\mathrm{T}} \mathbf{E}_1^2 \mathbf{S}_1 + \mathbf{G}_1 \mathbf{G}_2 \mathbf{S}_1^{\mathrm{T}} \mathbf{E}_1 \mathbf{S}_1 + \mathbf{G}_2 \mathbf{G}_1 \mathbf{S}_1^{\mathrm{T}} \mathbf{E}_1 \mathbf{S}_1 ) + \frac{1}{| w_1 |^2 w_2^2} \mathbf{G}_2^2 \mathbf{S}_1^{\mathrm{T}} \mathbf{E}_1 \mathbf{S}_1  \}  \nonumber \\
    &= \frac{1}{| w_1 |} tr\{ (\frac{\mathbf{G}_1}{| w_1 |} - \frac{\mathbf{G}_2}{w_2})^2 \mathbf{S}_1^{\mathrm{T}} \frac{\mathbf{E}_1}{| w_1 |} \mathbf{S}_1  +  (\frac{\mathbf{G}_1}{| w_1 |} - \frac{\mathbf{G}_2}{w_2}) \mathbf{S}_1^{\mathrm{T}} (\frac{\mathbf{E}_1}{| w_1 |})^2 \mathbf{S}_1  +  \mathbf{S}_1^{\mathrm{T}} (\frac{\mathbf{E}_1}{| w_1 |})^3 \mathbf{S}_1  \}
\end{align}

Since the matrices
\begin{align*}
\frac{\mathbf{G}_1}{| w_1 |} - \frac{\mathbf{G}_2}{w_2}, \qquad \frac{\mathbf{E}_1}{| w_1 |}
\end{align*}
are real-value symmetric matrices and $\mathbf{E}_1$ is positive semi-definite as well according to (\ref{eq:vars_E1_pos}), we can decompose them as \cite{Horn1990}
\begin{align}  \label{eq:G1+G2+E1_decompose}
\frac{\mathbf{G}_1}{| w_1 |} - \frac{\mathbf{G}_2}{w_2} = \mathbf{U} \mathbf{X} \mathbf{U}^{\mathrm{T}}, \qquad \frac{\mathbf{E}_1}{| w_1 |} = \mathbf{V} \mathbf{Y} \mathbf{V}^{\mathrm{T}},
\end{align}
where the matrices $\mathbf{U}$ and $\mathbf{V}$ are orthonormal, and the matrices $\mathbf{X}$ and $\mathbf{Y}$ are diagonal with $\mathbf{Y}$ being positive semi-definite as well. Denote
\begin{align*}
\mathbf{X} \equiv \begin{bmatrix} x_1 & & & \\ & x_2 & & \\ & & \ddots & \\ & & & x_n \end{bmatrix}, \quad
\mathbf{Y} \equiv \begin{bmatrix} y_1 & & & \\ & y_2 & & \\ & & \ddots & \\ & & & y_n \end{bmatrix}, \quad
\mathbf{Z} \equiv \mathbf{V}^{\mathrm{T}} \mathbf{S}_1 \mathbf{U} \equiv \begin{bmatrix} z_{11} & z_{12} & \cdots & z_{1n} \\ z_{21} & z_{22} & \cdots & z_{2n} \\ \vdots & \vdots & \ddots & \vdots \\ z_{n1} & z_{n2} & \cdots & z_{nn} \end{bmatrix},
\end{align*}
where each $y_i \geq 0$.

Substitute (\ref{eq:G1+G2+E1_decompose}) into (\ref{eq:dif4_F2_new}) and obtain
\begin{align}  \label{eq:dif4_F2_new2}
F_2 &= \frac{1}{| w_1 |} tr\{ (\mathbf{U} \mathbf{X} \mathbf{U}^{\mathrm{T}})^2 \mathbf{S}_1^{\mathrm{T}} \mathbf{V} \mathbf{Y} \mathbf{V}^{\mathrm{T}} \mathbf{S}_1  +  \mathbf{U} \mathbf{X} \mathbf{U}^{\mathrm{T}} \mathbf{S}_1^{\mathrm{T}} (\mathbf{V} \mathbf{Y} \mathbf{V}^{\mathrm{T}})^2 \mathbf{S}_1  +  \mathbf{S}_1^{\mathrm{T}} (\mathbf{V} \mathbf{Y} \mathbf{V}^{\mathrm{T}})^3 \mathbf{S}_1  \}  \nonumber \\
  &= \frac{1}{| w_1 |} tr\{ \mathbf{U} \mathbf{X}^2 \mathbf{U}^{\mathrm{T}} \mathbf{S}_1^{\mathrm{T}} \mathbf{V} \mathbf{Y} \mathbf{V}^{\mathrm{T}} \mathbf{S}_1  +  \mathbf{U} \mathbf{X} \mathbf{U}^{\mathrm{T}} \mathbf{S}_1^{\mathrm{T}} \mathbf{V} \mathbf{Y}^2 \mathbf{V}^{\mathrm{T}} \mathbf{S}_1  +  \mathbf{U} \mathbf{U}^{\mathrm{T}} \mathbf{S}_1^{\mathrm{T}} \mathbf{V} \mathbf{Y}^3 \mathbf{V}^{\mathrm{T}} \mathbf{S}_1  \}  \nonumber \\
  &= \frac{1}{| w_1 |} tr\{ \mathbf{X}^2 \mathbf{U}^{\mathrm{T}} \mathbf{S}_1^{\mathrm{T}} \mathbf{V} \mathbf{Y} \mathbf{V}^{\mathrm{T}} \mathbf{S}_1 \mathbf{U}  +  \mathbf{X} \mathbf{U}^{\mathrm{T}} \mathbf{S}_1^{\mathrm{T}} \mathbf{V} \mathbf{Y}^2 \mathbf{V}^{\mathrm{T}} \mathbf{S}_1 \mathbf{U}  +  \mathbf{U}^{\mathrm{T}} \mathbf{S}_1^{\mathrm{T}} \mathbf{V} \mathbf{Y}^3 \mathbf{V}^{\mathrm{T}} \mathbf{S}_1 \mathbf{U} \}  \nonumber \\
  &= \frac{1}{| w_1 |} tr\{ \mathbf{X}^2 \mathbf{Z}^{\mathrm{T}} \mathbf{Y} \mathbf{Z}  +  \mathbf{X} \mathbf{Z}^{\mathrm{T}} \mathbf{Y}^2 \mathbf{Z}  +  \mathbf{Z}^{\mathrm{T}} \mathbf{Y}^3 \mathbf{Z} \} = \frac{1}{| w_1 |} \sum_{i,j=1}^n (x_i^2 + x_i y_j + y_j^2) y_j z_{ji}^2 \geq 0.
\end{align}
By symmetry between $F_3$ and $F_2$ we can also prove $F_3 \geq 0$ in similar way. So the inequality (\ref{eq:dif4_F2+F3}) is proved.

\subsection{Positive semi-definiteness of $F_{41}$ and $F_{42}$}  \label{app:dif4_F41+F42}

Recall the second equation of (\ref{eq:vars_Q=Pinv}), the first equation of (\ref{eq:vars_P3related}), and (\ref{eq:vars_S1+E1}). Consider $F_{41}$ of (\ref{eq:lndetP_4thDif_decompose}) and we have
\begin{align*}
F_{41} &= tr\{ \frac{1}{w_1^4} [ \mathbf{R}_1 \mathbf{D}_1 (\mathbf{I} - \mathbf{Q}_1 \mathbf{D}_1) ]^2 \} = \frac{1}{w_1^4} tr\{ [ \mathbf{Q}_1 \mathbf{Q}_3^{-1} \mathbf{Q}_1 \mathbf{D}_1 (\mathbf{I} - \mathbf{Q}_1 \mathbf{D}_1) ]^2 \}  \\
  &= \frac{1}{w_1^4} tr\{ \mathbf{Q}_1 \sqrt{\mathbf{Q}_3^{-1}} \sqrt{\mathbf{Q}_3^{-1}} \mathbf{Q}_1 \sqrt{\mathbf{D}_1} (\mathbf{I} - \sqrt{\mathbf{D}_1} \mathbf{Q}_1 \sqrt{\mathbf{D}_1}) \sqrt{\mathbf{D}_1}  \\ 
  & \qquad \qquad \mathbf{Q}_1 \sqrt{\mathbf{Q}_3^{-1}} \sqrt{\mathbf{Q}_3^{-1}} \mathbf{Q}_1 \sqrt{\mathbf{D}_1} (\mathbf{I} - \sqrt{\mathbf{D}_1} \mathbf{Q}_1 \sqrt{\mathbf{D}_1}) \sqrt{\mathbf{D}_1} \}  \\
  &= \frac{1}{w_1^4} tr\{ \sqrt{\mathbf{Q}_3^{-1}} \mathbf{Q}_1 \sqrt{\mathbf{D}_1} (\mathbf{I} - \sqrt{\mathbf{D}_1} \mathbf{Q}_1 \sqrt{\mathbf{D}_1}) \sqrt{\mathbf{D}_1} \mathbf{Q}_1 \sqrt{\mathbf{Q}_3^{-1}}  \\ 
  & \qquad \qquad \sqrt{\mathbf{Q}_3^{-1}} \mathbf{Q}_1 \sqrt{\mathbf{D}_1} (\mathbf{I} - \sqrt{\mathbf{D}_1} \mathbf{Q}_1 \sqrt{\mathbf{D}_1}) \sqrt{\mathbf{D}_1} \mathbf{Q}_1 \sqrt{\mathbf{Q}_3^{-1}} \}  \\
  &= \frac{1}{w_1^4} tr\{ (\mathbf{S}_1^{\mathrm{T}} \mathbf{E}_1 \mathbf{S}_1)^{\mathrm{T}} (\mathbf{S}_1^{\mathrm{T}} \mathbf{E}_1 \mathbf{S}_1) \} \geq 0,
\end{align*}
because $(\mathbf{S}_1^{\mathrm{T}} \mathbf{E}_1 \mathbf{S}_1)^{\mathrm{T}} (\mathbf{S}_1^{\mathrm{T}} \mathbf{E}_1 \mathbf{S}_1)$ is positive semi-definite. By symmetry between $F_{42}$ and $F_{41}$ we can also prove $F_{42} \geq 0$ in similar way. So the inequality (\ref{eq:dif4_F41+F42}) is proved.

\subsection{Positive semi-definiteness of $F_{43}$}  \label{app:dif4_F43}

Besides the notations defined in (\ref{eq:vars_S1+E1}), similarly define the following notations
\begin{equation}  \label{eq:vars_S2+E2}
\mathbf{S}_2 \equiv \sqrt{\mathbf{D}_2} \mathbf{Q}_2 \sqrt{\mathbf{Q}_3^{-1}}, \qquad \mathbf{E}_2 \equiv \mathbf{I} - \sqrt{\mathbf{D}_2} \mathbf{Q}_2 \sqrt{\mathbf{D}_2}.
\end{equation}
Like (\ref{eq:vars_E1_pos}), we also have
\begin{align}  \label{eq:vars_E2_pos}
\mathbf{E}_2 = \sqrt{\mathbf{D}_2} \mathbf{D}_2^{-1} \sqrt{\mathbf{D}_2} - \sqrt{\mathbf{D}_2} \mathbf{P}_2^{-1} \sqrt{\mathbf{D}_2} = \sqrt{\mathbf{D}_2} (\mathbf{D}_2^{-1} - \mathbf{P}_2^{-1}) \sqrt{\mathbf{D}_2} \geq 0.
\end{align}
Recall the first equation of (\ref{eq:vars_P3related}). Then consider $F_{43}$ of (\ref{eq:lndetP_4thDif_decompose}) and we have
\begin{align*}
F_{43} &= tr\{ \frac{1}{w_1^2 w_2^2} \mathbf{P}_3^{-1} \mathbf{D}_1 (\mathbf{I} - \mathbf{Q}_1 \mathbf{D}_1) \mathbf{P}_3^{-1} \mathbf{D}_2 (\mathbf{I} - \mathbf{Q}_2 \mathbf{D}_2) \}  \\
  &= \frac{1}{w_1^2 w_2^2} tr\{ \mathbf{Q}_2 \mathbf{Q}_3^{-1} \mathbf{Q}_1 \mathbf{D}_1 (\mathbf{I} - \mathbf{Q}_1 \mathbf{D}_1) \mathbf{Q}_1 \mathbf{Q}_3^{-1} \mathbf{Q}_2 \mathbf{D}_2 (\mathbf{I} - \mathbf{Q}_2 \mathbf{D}_2) \}  \\
  &= \frac{1}{w_1^2 w_2^2} tr\{ \mathbf{Q}_2 \sqrt{\mathbf{Q}_3^{-1}} \sqrt{\mathbf{Q}_3^{-1}} \mathbf{Q}_1 \sqrt{\mathbf{D}_1} (\mathbf{I} - \sqrt{\mathbf{D}_1} \mathbf{Q}_1 \sqrt{\mathbf{D}_1}) \sqrt{\mathbf{D}_1}  \\
  & \qquad \qquad \quad \mathbf{Q}_1 \sqrt{\mathbf{Q}_3^{-1}} \sqrt{\mathbf{Q}_3^{-1}} \mathbf{Q}_2 \sqrt{\mathbf{D}_2} (\mathbf{I} - \sqrt{\mathbf{D}_2} \mathbf{Q}_2 \sqrt{\mathbf{D}_2}) \sqrt{\mathbf{D}_2} \}  \\
  &= \frac{1}{w_1^2 w_2^2} tr\{ \sqrt{\mathbf{Q}_3^{-1}} \mathbf{Q}_1 \sqrt{\mathbf{D}_1} (\mathbf{I} - \sqrt{\mathbf{D}_1} \mathbf{Q}_1 \sqrt{\mathbf{D}_1}) \sqrt{\mathbf{D}_1} \mathbf{Q}_1 \sqrt{\mathbf{Q}_3^{-1}}  \\
  & \qquad \qquad \quad \sqrt{\mathbf{Q}_3^{-1}} \mathbf{Q}_2 \sqrt{\mathbf{D}_2} (\mathbf{I} - \sqrt{\mathbf{D}_2} \mathbf{Q}_2 \sqrt{\mathbf{D}_2}) \sqrt{\mathbf{D}_2} \mathbf{Q}_2 \sqrt{\mathbf{Q}_3^{-1}} \}  \\
  &= \frac{1}{w_1^2 w_2^2} tr\{ \mathbf{S}_1^{\mathrm{T}} \mathbf{E}_1 \mathbf{S}_1 \mathbf{S}_2^{\mathrm{T}} \mathbf{E}_2 \mathbf{S}_2 \} = \frac{1}{w_1^2 w_2^2} tr\{ \sqrt{\mathbf{S}_1^{\mathrm{T}} \mathbf{E}_1 \mathbf{S}_1}^{\mathrm{T}} \mathbf{S}_2^{\mathrm{T}} \mathbf{E}_2 \mathbf{S}_2 \sqrt{\mathbf{S}_1^{\mathrm{T}} \mathbf{E}_1 \mathbf{S}_1} \}  \geq 0.
\end{align*}
So the inequality (\ref{eq:dif4_F43}) is proved.

%%%%%%%%%%%%%%%%%%%%%%%%%%%%%%%%%%%%%%%%%%%%%%%%%%%%%%%%

\newpage
\addcontentsline{toc}{chapter}{Bibliography}

\fancyhf{} % clear header/footer

\bibliographystyle{unsrt}
\bibliography{LI_Hao_Refs_EFISCIF}

\fancyhead[LE,RO]{\thepage}
\fancyhead[RE]{\textit{ \nouppercase{\leftmark}} }
\fancyhead[LO]{\textit{ \nouppercase{\rightmark}} }

\end{document}